\documentclass[aps,prl,twocolumn,showpacs,superscriptaddress]{revtex4-2}

\usepackage[utf8]{inputenc}
\usepackage[english]{babel}

\usepackage{amsthm}
\usepackage{mathrsfs}
\usepackage{amssymb} 
\newtheorem{theorem}{Theorem}
\usepackage[letterpaper,top=2cm,bottom=2cm,left=3cm,right=3cm,marginparwidth=1.75cm]{geometry}

\usepackage{amsmath}
\usepackage{graphicx}
\usepackage[colorlinks=true, allcolors=blue]{hyperref}

\usepackage{url}
\usepackage{braket}

\usepackage{xcolor}

\usepackage{subcaption}

\begin{document}

\title{Error-mitigated photonic quantum circuit Born machine}
\author{Alexia Salavrakos}
\email{alexia.salavrakos@quandela.com}
\affiliation{Quandela, 7 Rue Léonard de Vinci, 91300 Massy, France}
\author{Tigran Sedrakyan} 
\affiliation{Quandela, 7 Rue Léonard de Vinci, 91300 Massy, France}
\author{James Mills}
\affiliation{Quandela, 7 Rue Léonard de Vinci, 91300 Massy, France}
\affiliation{School of Informatics, University of Edinburgh, 10 Crichton Street, Edinburgh EH8 9AB, Scotland}
\author{Shane Mansfield}
\affiliation{Quandela, 7 Rue Léonard de Vinci, 91300 Massy, France}
\author{Rawad Mezher}
\affiliation{Quandela, 7 Rue Léonard de Vinci, 91300 Massy, France}

\begin{abstract}
In this article, we study quantum circuit Born machines (QCBMs) in the context of photonic quantum computing. QCBMs are a popular choice of quantum generative machine learning models, and we present a QCBM designed for linear optics. We show that a recently developed error mitigation technique called recycling mitigation greatly improves the training of QCBMs in realistic scenarios with photon loss, which is the primary source of noise in photonic systems. We demonstrate this through numerical simulations and through an experiment on a quantum photonic integrated processor. We expect our work to pave the way towards more demonstrations of error mitigation techniques tailored to photonic devices which can enhance the performance of a quantum algorithm.
\end{abstract}


\maketitle

\section{Introduction}
Generative learning has captured much attention in the field of classical machine learning over the last decade, with the advent of deep generative models like variational autoencoders (VAEs), generative adversarial networks (GANs), diffusion models, and most recently large auto-regressive models \cite{Kingma2014, Goodfellow2014, pmlr-v37-sohl-dickstein15}. Over the same period of time, the field of quantum machine learning has emerged and grown fast, with a recent focus on variational quantum algorithms \cite{Cerezo_2021}, which are hybrid quantum-classical models that can be implemented on noisy-intermediate scale quantum (NISQ) devices \cite{Preskill_2018}. In these algorithms, the model is implemented by a parametrized quantum circuit often called \emph{ansatz}, and the parameters are optimized through a classical procedure. While classification tasks have been studied extensively \cite{Schuld_2020, farhi2018}, considerable progress has also been achieved on quantum generative models \cite{Tian_2022}. Quantum GANs and quantum VAEs were introduced in \cite{Lloyd_2018, Dallaire_Demers_2018} and \cite{Khoshaman_2018} respectively, while energy based models such as quantum Boltzmann machines were studied in \cite{Amin_2018, wiebe2019generative}.


Quantum Circuit Born Machines (QCBMs) were introduced as a circuit-based model \cite{Benedetti_2019, Liu_2018}, where the circuit ansatz prepares a state and the Born rule is naturally implemented by the measurements at the end of the circuit. These models are implementable on NISQ hardware, and a first experimental realization was obtained in \cite{Benedetti_2019} on an ion trap quantum processor. Since then, further demonstrations of QCBMs and QCBM-like models were realized on ion traps \cite{Zhu2019, Rudolph_2022}, as well as on superconducting platforms \cite{Hamilton2019, Coyle_2021, leytonortega2021}. 

Each platform that is being considered for the development of a quantum computer has its strengths and weaknesses, whose importance may vary according to the task at hand. Fast sampling, for instance, could be particularly interesting for generative learning, and is a benefit of photon-based devices. Quandela introduced in \cite{maring2023} a quantum processor called \emph{Ascella}, made of a quantum-dot single-photon source supplying a universal linear optical network on a reconfigurable chip. A very similar processor with a larger chip called \emph{Altair} was recently developed \cite{QuandelaCloud}. Photonic platforms are also known for demonstrations of quantum advantage \cite{Zhong_2020, Madsen2022} via Gaussian boson sampling \cite{Hamilton_2017}.

While such devices are currently state-of-the-art, they have not yet entered a regime where quantum error correction (QEC) can be applied at scale.  In contrast to QEC protocols, which require a large overhead of qubits to be realised \cite{Lidarbook}, quantum error mitigation (QEM) replaces this qubit overhead with a sampling overhead.  More precisely,  QEM requires more runs of a quantum algorithm, but allows this algorithm to be run on a NISQ device without the need for any ancillary qubits \cite{Endo_2018,QEMreview}.
This makes QEM particularly useful for current and near-term quantum hardware.
Photonic devices suffer significantly from a particular type of noise: erasure noise in the form of photon loss. In related work \cite{MezherMills2024}, a recycling mitigation technique is presented that deals with photon loss. This technique makes use of photonic output states that would normally be discarded in post-selection, and provably outperforms post-selection up to large sample sizes.

In this work, after introducing a QCBM scheme for linear optical circuits, we show that its training is significantly improved by the recycling mitigation technique in realistic scenarios with photon loss. We demonstrate several situations where a QCBM becomes trainable if recycling mitigation is applied, including a proof-of-concept experiment on our photonic quantum processor \emph{Altair}. 
In the discussion, we consider the question of classical simulability, hardness, and the link with boson sampling \cite{aaronson2011,berkowitz2018stability}.

\begin{figure*}
	\centering
	\begin{subfigure}[c]{0.6\linewidth}
		\includegraphics[width=\linewidth]{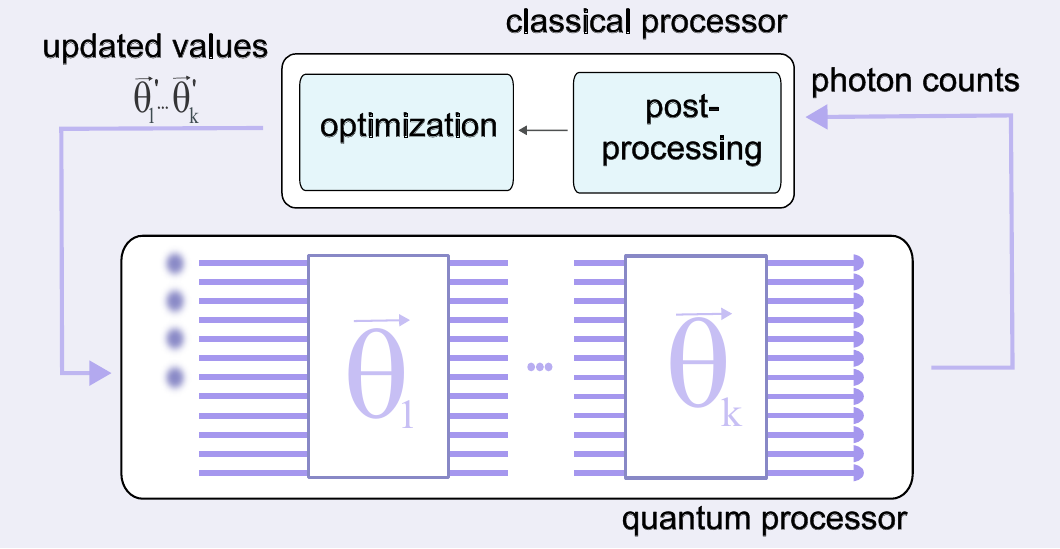}
		\caption{}
		\label{fig:ansatzsubfigA}
	\end{subfigure}
	\begin{subfigure}[c]{0.25\linewidth}
		\includegraphics[width=\linewidth]{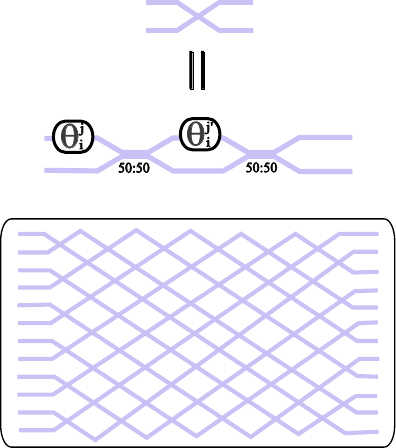}
		\caption{}
		\label{fig:ansatzsubfigB}
	\end{subfigure}
	\caption{(a) QCBM ansatz pictured for a  circuit with  $m = 12$ modes and $n =4$ photons and input Fock state $\ket{101010100000}$. 
 We allow the variational blocks parametrized by $\vec{\theta_i}$ to be repeated $k$ times. 
 The photon coincidence counts measured at the end of the circuit are sent to a classical processor. 
 (b) One variational block is a universal interferometer. Each crossing between the optical modes corresponds to a Mach-Zender interferometer that is parametrized by two phases.} 
	\label{qcbm_ansatz}
\end{figure*}

\section{A photonic QCBM}
\emph{Generative learning.} 
Given a set of data instances $X$ and labels $Y$, while a discriminative or classification model aims at predicting $P(Y|X)$, a generative model seeks to estimate the joint distribution $P(X, Y)$, or simply $P(X)$ if there are no labels. If the model has parameters $\vec{\theta}$, a loss function $\mathcal{L}$ is defined which should express how close the distribution of the generated samples is to the one of X.  The loss is optimized so as to obtain $\vec{\theta}^* = \arg \min_{\vec{\theta}} \mathcal{L}(\theta|X)$.
Once the model is trained, it can be used to generate new samples with similar properties to the training dataset.

\emph{Photonic architecture.} We propose using an ansatz based on the universal interferometer design of \cite{Clements_2016}, as shown in Figure \ref{qcbm_ansatz}.  We believe our choice to be well motivated, since the interferometer guarantees interactions between all modes, allows for arbitrary unitary transformation on the modes, and matches the scenario of boson sampling. Note that in \cite{shankar2022variational, Shi2023}, the authors present QCBMs with a different linear optical ansatz 
resembling the gate-based one of \cite{Liu_2018}.

Our QCBM ansatz consists of a linear optical circuit with $m$ modes and $n$ photons. It receives a Fock state as input: $\ket{\Vec{n}_{in}} = \ket{n^{in}_1, ..., n^{in}_m}$, where $n^{in}_i$ indicates the number of photons in mode $i$ and $\sum_{i} n^{in}_i = n$. This state is transformed according to the parametrized universal interferometer which implements a unitary $U(\Vec{\theta})$. Indeed, the scheme of \cite{Clements_2016} allows for the decomposition of an arbitrary unitary into a product of $T(\theta^{j}_{i}, \theta'^{j}_{i})$ matrices. Each matrix corresponds to a Mach-Zender element between the  $j$-th and $(j+1)$-th modes, as depicted in Fig \ref{fig:ansatzsubfigB}, and has the form:

\begin{equation*}
    \begin{pmatrix}
1&0&\cdots&&&\cdots&&0\\
0&1&&&&&&\vdots\\
\vdots&&&e^{i\theta^{j}_{i}} \cos{\frac{\theta'^{j}_{i}}{2}}&-\sin{\frac{\theta'^{j}_{i}}{2}}&\\
&&&e^{i\theta^{j}_{i}} \sin{\frac{\theta'^{j}_{i}}{2}}&\cos{\frac{\theta'^{j}_{i}}{2}}&&&\vdots\\
\vdots&&&&&&1&0\\
0&\cdots&&&&\cdots&0&1\\
\end{pmatrix}
\end{equation*}
There are $m(m-1)/2$ Mach-Zender elements per universal interferometer as depicted in Figure \ref{fig:ansatzsubfigB}, and thus $m(m-1)$ phases or parameters per block. While a sequence of unitaries can always be rewritten as one single unitary and thus implemented with one block, we consider it interesting to explore whether this parametrization structure can impact training performance, along with other parametrization choices such as setting all $\theta^{j}_{i}$ to a fixed value and optimizing only over the $\theta'^{j}_{i}$. In practice however, we obtained the best numerical results for $k = 1$.

At the output of the circuit, the state is measured by photon detectors. Different Fock states $\ket{\Vec{n}_{out}} = \ket{n^{out}_1, ..., n^{out}_m}$ are detected as arrangements of photons in the output modes, that we denote $s 
 = (n^{out}_1, ..., n^{out}_m)$. With photon loss, the $n^{out}_i$ may not sum to $n$. Detecting all possible output states is not straightforward in practice because photon number resolving detectors are not readily available with current technology. To match experimental settings, we make the simplifying choice of using threshold detectors in most of our simulations, thus only allowing $n^{out}_i$ to be $0$ or $1$ (`click' or `no click'). We refer to the Supplemental Material \cite{supplemental} for more details on no-collision regimes and recycling mitigation. 
 
 After detection, outputs $s$ are mapped through classical post-processing to the space of the data that the model aims to generate. For instance, if we consider simple one-dimensional numerical data $X$ set between $x_{\text{min}}$ and $x_{\text{max}}$, we proceed by binning the range of values $[x_{\text{min}}, x_{\text{max}}]$, then defining a mapping between each possible output and each bin. If the data is discrete, there is no need for the binning, and a certain output $s^*$ can directly be interpreted as a value $x^*$.

\emph{Training strategy.} The model is trained as a variational algorithm, which means that a classical optimization procedure is applied to the model parameters $\Vec{\theta}$, as shown in Figure \ref{qcbm_ansatz}. 
In many of our simulations, we use the Kullback-Leibler (KL) divergence \cite{Kullback1951} as a loss function. We recall the definition of the KL divergence, given two probability distributions $Q$ and $P$ defined on the same sample space $\mathcal{X}$:
\begin{equation}
    KL (P || Q)= \sum_{x \in \mathcal{X}} P(x) \text{log}\left(\frac{P(x)}{Q(x)}\right).
\end{equation}
We also use the Total Variational Distance (TVD) which can be written on discrete domains as \cite{levin2017markov}:
\begin{equation}
    \text{TVD} (P, Q)= \frac{1}{2} \sum_{x \in \mathcal{X}}|P(x) - Q(x)|.
\end{equation}
Both the KL and the TVD measure how different $Q$ is from $P$. The lower the value of the metric, the more similar the distributions are, and a value of 0 can reached if the two distributions are identical. This is not always desirable as it may indicate overfitting or memorization of the model, and slightly higher values could imply better generalization properties as pointed out in \cite{Gili_2023}. 

In generative learning, additional metrics are often used outside of the loss function to evaluate new model samples, sometimes with respect to a test set. For the generation of images, examples include the Fréchet inception distance \cite{fid} or the structural similarity index measure \cite{ssim}. We consider this to be outside the scope of this article, as we focus on the impact of recycling mitigation on loss function behaviour.

Other popular choices of loss functions for quantum models include the Maximum Mean Discrepancy (MMD) \cite{Gretton_2012} or the Sinkhorn divergence studied in \cite{Coyle_2020}. We refer to the last section of this paper for a discussion around the choice of the loss function, and to the Supplemental Material \cite{supplemental} for more details on model training.

\section{Recycling mitigation}
Recycling mitigation is a technique for mitigating the effects of uniform photon loss in linear optical quantum circuits consisting of single photons, multi-mode linear optical interferometers, and single photon detectors \cite{MezherMills2024}. At the heart of recycling mitigation are the \emph{recycled probabilities}, constructed from the statistics of the linear optical circuit affected by photon loss, that is, statistics where the number of detected  photons at the output of the circuit is less than 
 at the input. The standard method in near-term photonic quantum algorithms is to apply post-selection, i.e. discard the statistics where photon loss occurs, and only use the statistics where no loss occurred for the computation. 

For a given output $s$ of a linear optical circuit, as defined above, the recycled  probability $p_R(s)$ is of the form
\begin{equation}
   p_R(s)=p_1p_{id}(s)+(1-p_1)I(s),
\end{equation}
where $0\leq p_1 \leq 1$, $p_{id}(s)$ is the ideal probability of outcome $s$ in the absence of photon loss, and $I(s)$ is the  \emph{interference term}, which is an artifact of constructing recycled probabilities from probabilities affected by photon loss. Recycling mitigation consists of applying certain classical post-processing techniques, detailed in \cite{MezherMills2024}, to the set of recycled probabilities $\{p_R(s)\}$, resulting in a set of mitigated probabilities $\{p_{mit}(s)\}$, which are approximations of the ideal probabilities $\{p_{id}(s)\}$.

On an intuitive level, recycling mitigation uses lossy photon statistics in order to get  estimates of the ideal (lossless) probabilities. Crucially, the goal of recycling mitigation is to beat post-selection, in the sense of getting lower statistical errors on the estimated ideal probabilities. Post-selection is an unbiased estimator of the ideal probability with the only error being a statistical error $\epsilon_{stat,post}$. Conversely, recycling mitigation results in a \emph{biased} estimator of the ideal probability, with a bias error $\epsilon_{bias}$ as well as a statistical error $\epsilon_{stat,rec}$. Because recycling mitigation uses the (more probable) lossy statistics to construct its estimators, then $\epsilon_{stat,rec} < \epsilon_{stat,post}$ by construction. Furthermore, it was shown in \cite{MezherMills2024} that the bias error is small enough such that the inequality 
\begin{equation}
\epsilon_{stat,rec}+\epsilon_{bias} \leq \epsilon_{stat,post},
\end{equation}
holds for up to a very large  number of samples, and for a generic choice of interferometer. This makes recycling mitigation a powerful error mitigation technique applicable to near-term photonic quantum hardware.


\section{Numerical simulations}
\begin{figure*}
	\centering
	\begin{subfigure}{0.3\linewidth}
		\includegraphics[width=\linewidth]{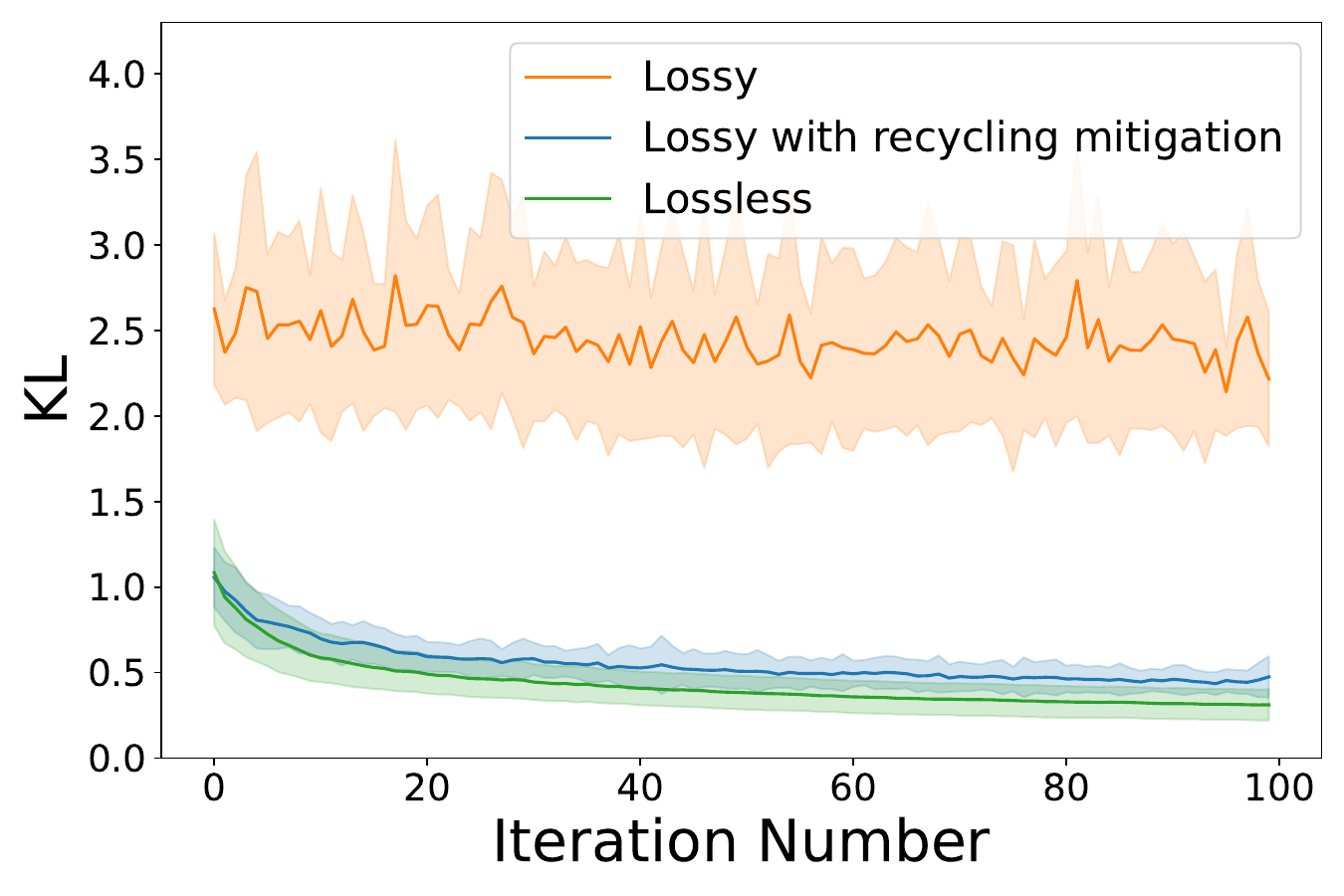}
		\caption{$n = 4$, $m = 12$, $N = 200000$, $\eta = 0.8$}
		\label{fig:subfigA}
	\end{subfigure}
	\begin{subfigure}{0.3\linewidth}
	        \includegraphics[width=\linewidth]{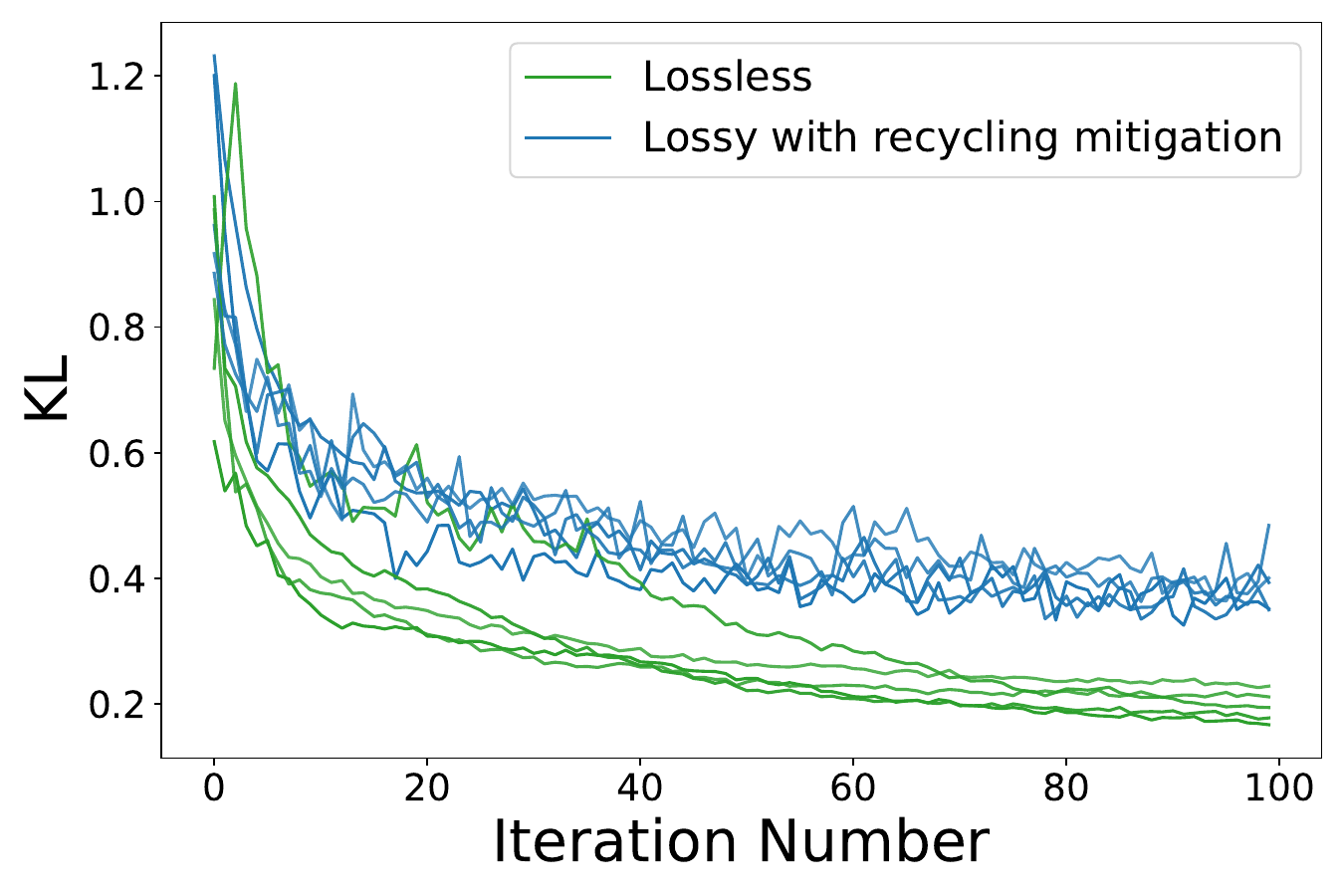}
	        \caption{$n = 4$, $m = 12$, $N = 200000$, $\eta = 0.8$}
	        \label{fig:subfigC}
         \end{subfigure}
         \begin{subfigure}{0.3\linewidth}
	        \includegraphics[width=\linewidth]{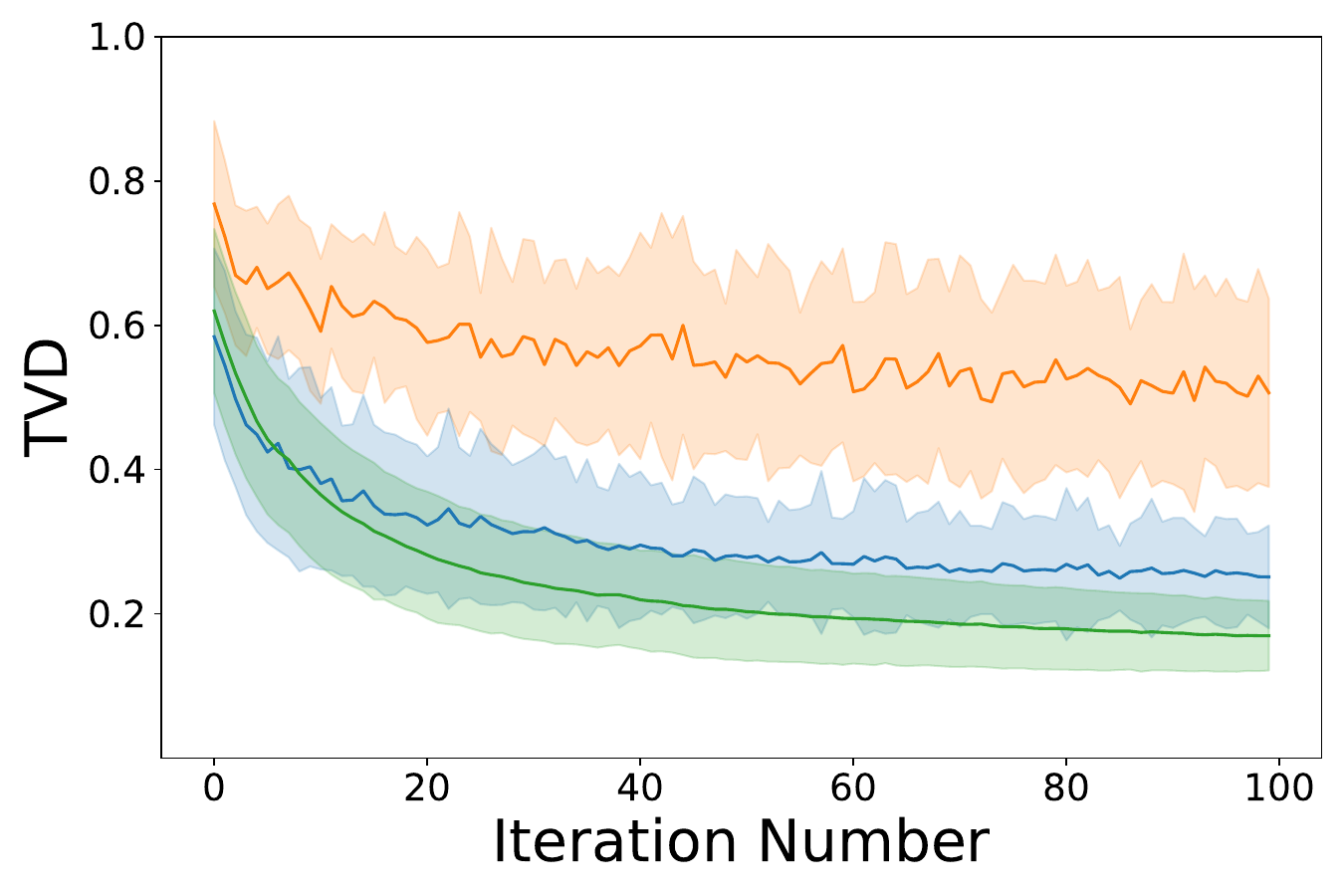}
	        \caption{$n = 4$, $m = 12$, $N = 60000$, $\eta = 0.8$}
	        \label{fig:subfigD}
         \end{subfigure}
        \begin{subfigure}{0.3\linewidth}
	        \includegraphics[width=\linewidth]{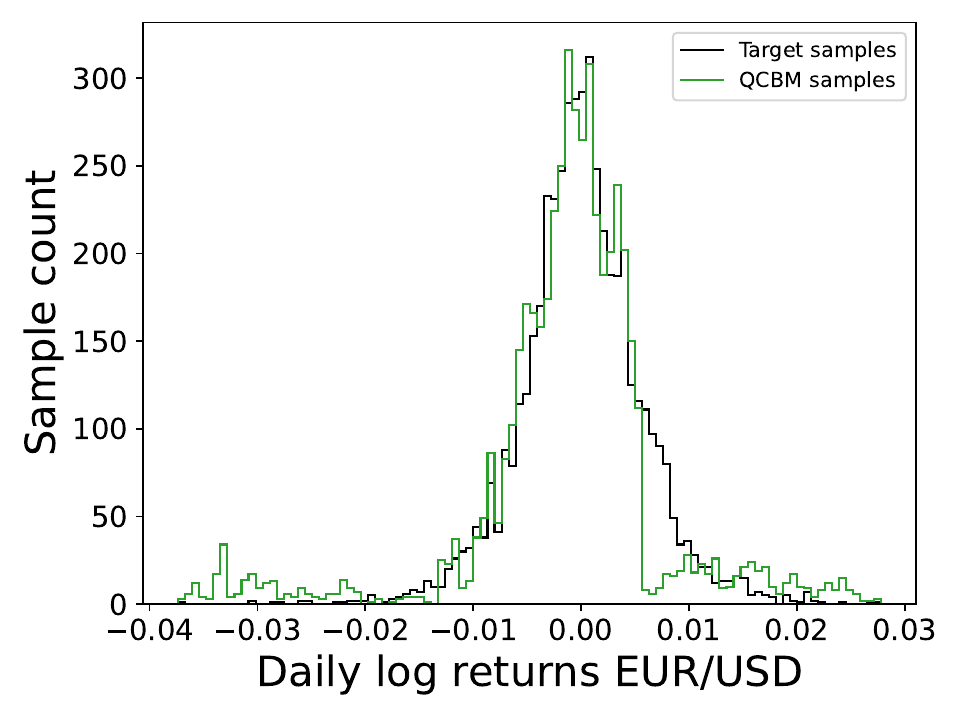}
	        \caption{$n = 4$, $m = 12$, $N = 60000$, $\eta = 0.8$}
	        \label{fig:subfigsamples}
         \end{subfigure}
        \begin{subfigure}{0.3\linewidth}
	        \includegraphics[width=\linewidth]{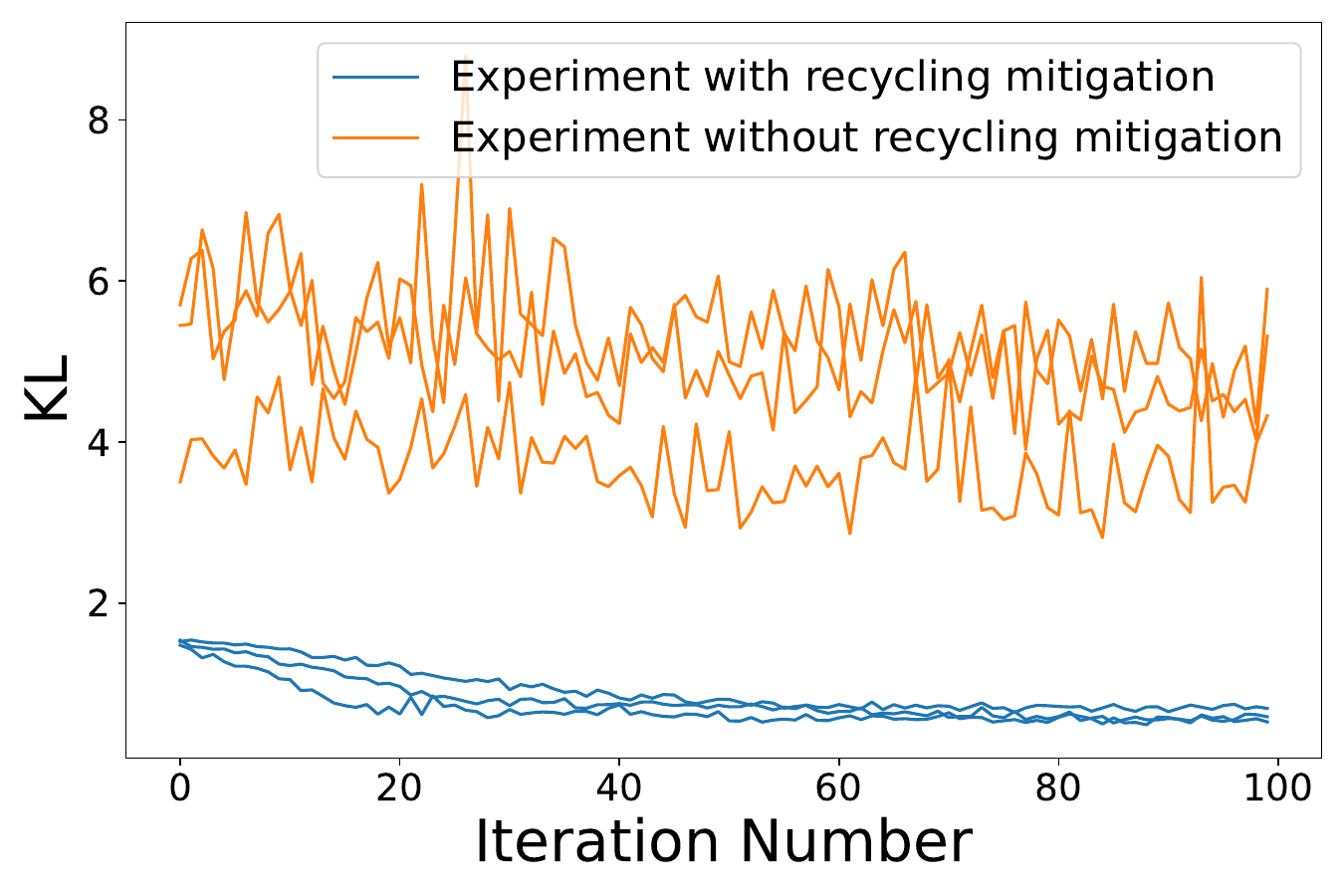}
	        \caption{$n = 4$, $m = 12$, $N = 5\times10^7$, $\eta \approx 0.96$}
	        \label{fig:subfigE}
         \end{subfigure}
         \begin{subfigure}{0.3\linewidth}
	        \includegraphics[width=\linewidth]{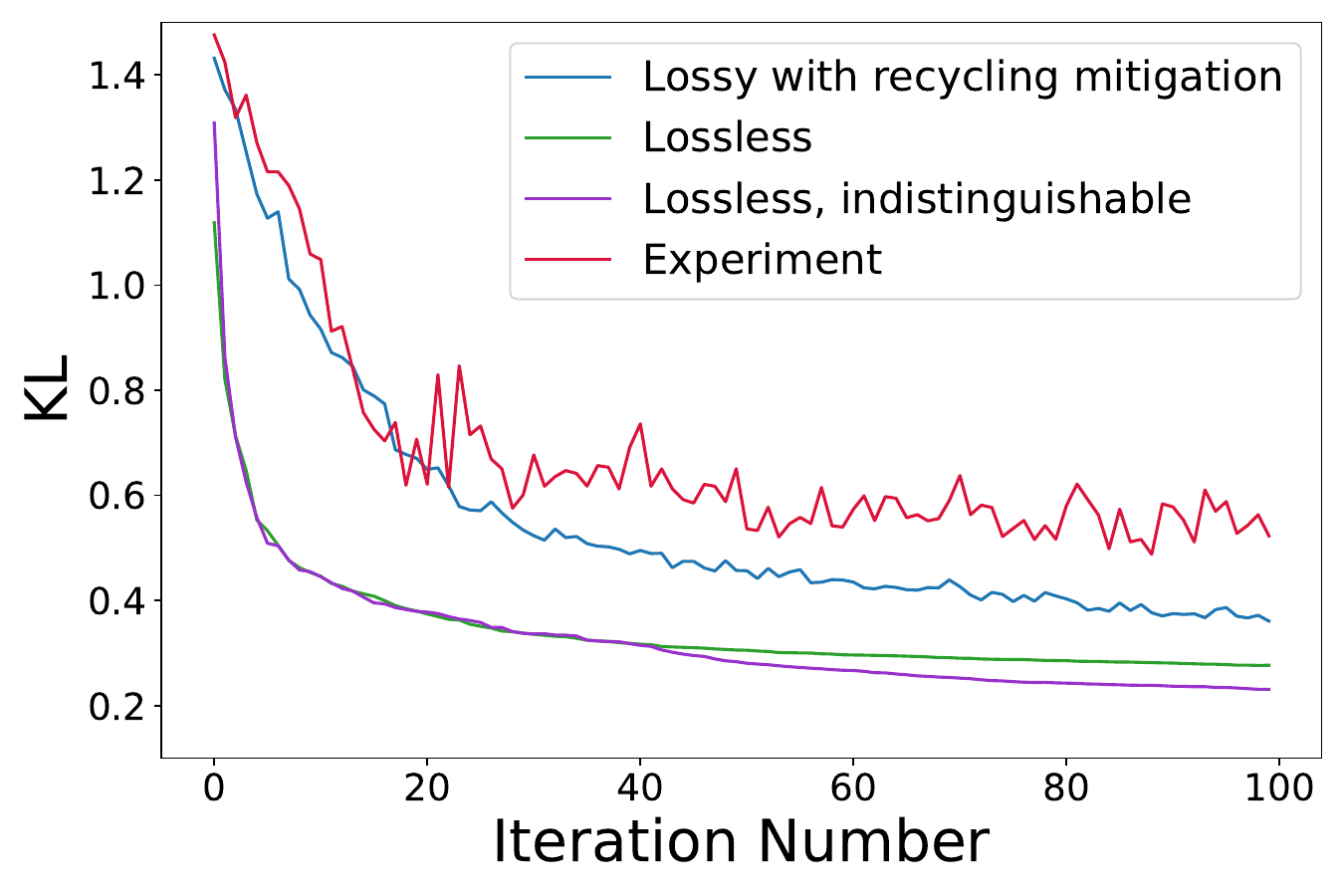}
	        \caption{$n = 4$, $m = 12$, $N = 5\times10^7$, $\eta \approx 0.96$}
	        \label{fig:subfigF}
         \end{subfigure}
	\caption{(a) Value of the KL loss function for the mixed Gaussian case, averaged over 20 training instances. The shaded area depicts a standard deviation from the mean. Each training instance has different initialization parameters. (b) We highlight the 5 best training instances out of the 20 simulations. (c) Value of the TVD loss function for the financial foreign exchange dataset. (d) Comparison between the target samples from the financial dataset and the samples generated by one of the trained (lossless) QCBM models. (e) Experimental results on \emph{Altair} for the mixed Gaussian case: three training instances with recycling mitigation, and three without. (f) Comparison between best mitigated experimental instance and corresponding simulations. All but one curve are affected by photon distinguishability.}
	\label{simulations}
\end{figure*}

\emph{Perceval} \cite{Heurtel_2023} is a software platform for discrete-variable photonic quantum computing, which we use in our simulations. 
We denote by $N$ the number of initial samples for a given simulation, and by $\eta$ the loss parameter, i.e. the probability that one photon is lost during the computation. Note that, when working with $n$ photons, the probability of observing $n$-photon coincidence counts at the output of the circuit scales as $(1-\eta)^n$. 

Given the same resources, we typically compare the following cases: (1) the ideal lossless case; (2) the lossy case, where for typical values of $\eta$, only a small fraction of $N$ are detected as $n$-photon counts, and where the $n'$-photon counts with $n' < n$ are discarded through post-selection; and (3) the same lossy case where recycling mitigation is applied, which means that information from the $n'$-photon counts is taken into account to train the QCBM, as explained in the previous section.

We summarise our results in Figure \ref{simulations}. We consider in Figure \ref{fig:subfigA}-\ref{fig:subfigC} the problem of learning a bimodal distribution consisting of two superimposed Gaussian distributions with different means, as in \cite{zoufal2021generative}. This distribution can be generated straightforwardly in Python. 
We then consider financial foreign exchange data in Figure \ref{fig:subfigD} -- as in \cite{Coyle_2021}, we collected daily log returns of the EUR/USD currency pair over the last 20 years to form the dataset. There, we use the TVD as a loss function as we found it more suitable for a sparser distribution. In Figure \ref{fig:subfigsamples}, we compare the distributions of the generated and target data.

We observe that for given $N$ and $\eta$, recycling mitigation significantly improves the training of lossy QCBMs and brings the value of the loss function close to the lossless one, while the lossy unmitigated case remains mostly untrainable. While the steadily decreasing trend of the loss function shows that the model is learning, we note that there is room for improvement in the model optimization, which would be interesting to study in future work. In particular, the importance of initialization parameters and the variance in the training process can be observed in Figure \ref{fig:subfigC}.  Note that simulation results are displayed for loss parameter $\eta = 0.8$, which is within reach for near-term hardware \cite{maring2023}. Our results remain valid for higher loss regimes and values of $\eta$, as we show in the next section.


\section{Experiment}
For our experimental setup, we use the \emph{Altair} quantum photonic integrated processor. We implement the same scenario as in Figure \ref{fig:subfigA}, although the main difference is that we find ourselves in a high loss regime: $\eta$ is estimated to be around 0.96. Additionally, the device is affected by other types of noise, such as photon distinguishability \cite{Shchesnovich2015, Tichy2015}, which would require a different kind of error mitigation \cite{Borzenkova_2024}. Further technical details on the setup can be found in the Supplemental Material \cite{supplemental}.

We present our results in Figures \ref{fig:subfigE} and \ref{fig:subfigF}. All runs were obtained with the same initialization, for a fairer comparison. In Figure \ref{fig:subfigE}, we observe as before a large gap between the unmitigated and mitigated cases. In Figure \ref{fig:subfigF}, we plot one mitigated experimental training instance along with corresponding simulations. We note the effect of distinguishability, as the fully noiseless simulation reaches a lower value of the loss. 

While the experimental data follows the simulation, there remains a difference in the final loss function value which we can attribute to the following factors: there are variations due to the stochasticity of the optimizer; in practice, losses are not fully uniform throughout the setup; and there may be inaccuracies in the transpilation, which is the process of applying adequate voltages to the chip to implement the right phase values, as explained in \cite{Fyrillas2024}.

As an alternative strategy for dealing with losses, we could have increased the initial number of samples $N$ enough to make the unmitigated case trainable. However, this would  significantly increase the time required for each experiment. 


\section{Discussion}
When considering quantum advantage of models based on linear optical circuits, it is natural to look for a connection with boson sampling. In this context, we note that generative models may be better candidates for a potential quantum advantage than classifiers obtained from a similar circuit, such as the ones in \cite{Gan_2022, maring2023}. Indeed, as Aaronson and Arkhipov pointed out in their original article \cite{aaronson2011}, if the circuit is used for probability estimation instead of sampling, there is a polynomial-time (in the input size) classical algorithm by Gurvits  which approximates the permanent of any submatrix of a unitary matrix to within polynomially-small additive error \cite{Gurvits_2005, aaronson2012}. This means that estimating the permanent from a polynomial (in input size) number of experiments does not yield an exponential advantage over classical algorithms estimating the permanent, by standard arguments from statistics \cite{hoeffding2012collected}. 

However, since generative models like QCBMs are based on sampling rather than permanent estimation, there remains potential for an exponential advantage. In \cite{Coyle_2020}, the authors clarified the distinction between classical hardness of simulating the circuits that make up the model, and what they defined as \emph{quantum learning supremacy}: the ability of a quantum model to provably outperform all classical models in a certain task. Our model has the capability to achieve classical hardness, provided $n$ and $m$ have sufficiently high values, and losses are limited \cite{Aaronson_2016}. Indeed, by varying the values of $\vec{\theta_k}$, the model can span across all possible $m$-mode interferometers, and there is very strong evidence that sampling, either exactly or approximately,  from the output distribution of  generic $m$-mode interferometers is hard to do classically \cite{aaronson2011, bouland2023complexitytheoreticfoundationsbosonsamplinglinear}. Quantum learning supremacy, however, remains an open question for this photonic QCBM as well as for the qubit-based one presented in \cite{Coyle_2020}.

While the above observations hold true when using a trained model to produce samples, they do not always apply to the training phase of the model. For instance, in our training, we exploited the output probability distribution of the circuit, which allowed us to apply recycling mitigation. The use of an explicit loss function \cite{rudolph2023trainability} based on the output probability distribution then comes at no additional cost, so it made sense for us to use the KL divergence or the TVD. 

Nevertheless, this particular approach to training retains potential, as it turns out that there is still a polynomial advantage in estimating probabilities given by the modulus squared of a permanent in a boson sampling experiment. Indeed, we show in the Supplemental Material \cite{supplemental} that statistics collected from $t$ i.i.d runs of a boson sampler allow, with high probability, for estimating the modulus squared of the permanent of an arbitrary matrix to a smaller additive error, as compared to the additive error  estimate  obtained from  Gurvits algorithm with a run time of $O(t)$. While these results hold for the lossless case only, this leads us to the open question: could recycling mitigation make these claims robust to noise? Finally, with respect to the  recent results of \cite{angrisani2024}, we point out that they do not apply to our setup as we do not consider expectation values of Pauli operators, but a linear optical evolution.

In future work, it would be interesting to extend recycling mitigation techniques to work with implicit loss functions like the MMD or the Sinkhorn divergence. These would allow for the training phase to be based on samples only, which would also be key to scaling the model \cite{rudolph2023trainability}. 
Other future research directions include developing mitigation techniques for other types of photonic errors, as well as improving model optimization in the context of photonics, focusing on topics such as gradient evaluation or parameter initialization, or considering circuit components beyond linear optics such as adaptivity and feedforward. 

\section{Conclusion}
Our work shows that recycling mitigation can positively impact the training of quantum machine learning models on realistic photonic devices. Given a certain amount of resources, for several scenarios recycling mitigation makes QCBMs that would otherwise be untrainable trainable. Our work also proposes a straightforward design for a photonic-tailored QCBM, which we hope can be built upon, as we make our code available \cite{repo}.

\section{Acknowledgements}
The authors would like to thank Eric Bertasi and Marion Fabre for their precious help improving the code, as well as Joseph Bowles and Nicolas Maring for fruitful discussions and feedback on the manuscript. This work has been co-funded by the European Commission as part of the EIC accelerator program under the grant agreement 190188855 for the SEPOQC project; by the Horizon-CL4 program under the grant agreement 101135288 for the EPIQUE project; and by the QuantERA project ResourceQ under Grant Agreement ANR-24-QUA2-007-003.

\bibliographystyle{apsrev4-1}
\bibliography{bibliography}
\nocite{Schuld_2019}
\nocite{vidal2018calculus}
\nocite{defelice2024}
\nocite{facelli2024}
\nocite{pappalardo2024}
\nocite{Spall_1998}
\nocite{Caro_2022}
\nocite{mezher2023solving}

\onecolumngrid
\appendix
\section{Supplemental Material}
\section{Recycling mitigation 
 and no-collision regime}\label{guarantees}
We note that the performance guarantees for recycling mitigation derived in [29] are for outputs $s$  where at most one photon occupies  a given mode. In the so-called no-collision regime where $m\gg n$ it is natural to expect such a behaviour of the outputs [30]. In our simulations and experiments however, $m$ and $n$ have a comparable size. Despite this, we observe a significant impact of the mitigation. In our simulations, we can map the simulated outputs $s$ where more than one photon can occupy a mode to threshold outputs, and use them in the recycling mitigation, whereas in post-selection they are discarded and the distribution is renormalised.

\section{Model training and hyperparameters}\label{model}
Gradient evaluation in quantum machine learning is less straightforward than in classical machine learning, due to the lack of an equivalent for the backpropagation algorithm. The parameter-shift method [58, 59] is the current standard for variational quantum circuits -- and only very recent work has focused on photonic quantum circuits [60-62]. Here, we use an optimizer based on the Simultaneous Perturbation Stochastic Approximation (SPSA) method [63]. It involves stochastic gradient approximation, thus allowing for fewer evaluations of the model at each step, which makes it practical for experiments. We start the optimization by randomly initializing the circuit parameters $\Vec{\theta}$.

During our training process, we compared different ansätze, or hyperparameters. We explored the case where $k > 1$, referring to Figure 1a of the main text, including the case where the parameters are repeated between blocks $\vec{\theta_1} = \vec{\theta_2} = ... = \vec{\theta_k}$. Interestingly, the repetition of gates has been shown to be useful for model properties in certain contexts [64]. We also tested the case where each Mach-Zender element between two modes in a block is parametrized by only one phase ($\theta'^{j}_{i}$ in Figure 1 of the main text), while the other one ($\theta^{j}_{i}$) is set to 0, thus halving the total number of parameters. Overall, we observed the best results for a simple ansatz made of a single block ($k = 1$), keeping both parameters per Mach-Zender element. This scenario corresponds to the results displayed in Figure 2 of the main text. As mentioned in the main text, a sequence of unitaries can be rewritten as a single unitary and thus implemented with one universal interferometer only. In that case, the choice $k > 1$ would correspond to a scenario where each phase in the circuit would be determined by more than one parameter of the model.

\section{Experimental setup}\label{experimental_setup}
The \emph{Altair} processor presents a similar architecture to the one presented in [21]. It consists of a single-photon source, a 10-mode active demultiplexer followed by fiber delays, a 20-mode photonic integrated circuit, and superconducting nanowire single-photon detectors (SNSPD). The single-photon source is a gated InGaAs quantum dot embedded inside a micropillar cavity. The quantum dot is excited by a laser to produce single photons on demand. At the time of the experiment, single-photon purity was estimated around $0.025 \pm 0.002$ and photon indistinguishability around $0.84 \pm 0.03$. The train of collected photons is then sent into the demultiplexer, which converts it into up to $n$ photons arriving simultaneously into different modes of the chip. The architecture of the chip is a universal interferometer (as in Figure 1b of the main text). Since the chip has $20$ modes, it contains $20 \times 19 = 380$ phases which are controlled and tuned thermo-optically. Note that a $12$-mode experiment can easily be implemented on a $20$-mode chip by injecting photons only into chosen modes among the first $12$, and blocking transmission into the last $8$ by controlling the phases. The SNSPD detectors are threshold detectors, which means that they can detect the presence of photons (`click' or `no click') and not the number of photons in each mode. The detection events are recorded via a time-to-digital converter. The \emph{Altair} processor is connected to a cloud platform [22], which allows specific operations to be sent to the processor, and which returns the results as detection statistics. 

\section{Approximating output probabilitities}\label{appendix_permanent}

In the context of boson sampling and classical hardness of the model training, we discuss approximating $|\mathsf{Per}(A)|^2$ for an arbitrary $n \times n$ matrix $A$ to a given additive error precision, $\mathsf{Per}(.)$ being the matrix permanent. We show that computing $|\mathsf{Per}(A)|^2$ from statistics collected from $t$ i.i.d runs of a boson sampler produces, with high probability, a more accurate estimate of $|\mathsf{Per}(A)|^2$
than Gurvits which is run for time $O(t)$. More precisely, we show in Theorems 1 and 2 that Gurvits algorithm requires a runtime $t^{*}$ with $t^{*}=O(n^2t)$ to produce an additive error approximation of $|\mathsf{Per}(A)|^2$ of the same order of magnitude  as that obtained from $O(t)$ runs of a boson sampler. 

Let $A$ be an $n \times n$ matrix with complex entries. Gurvits algorithm [50] provides an estimate $\mathsf{E}(\mathsf{Per}(A))$ of $\mathsf{Per}(A)$, the permanent of $A$, to within additive precision $\epsilon ||A||^n$, where $||A||$ is the spectral norm of $A$, and $\epsilon>0$. More precisely, the output of Gurvits algorithm is $\mathsf{E}(\mathsf{Per}(A))$ such that
$$|\mathsf{E}(\mathsf{Per}(A))-\mathsf{Per}(A)| \leq \epsilon ||A||^n.$$

The runtime of Gurvits is $O(\frac{n^2}{\epsilon^2})$, and the above inequality holds with high probability (by Hoeffding's inequality [52]). The $O(n^2)$ part comes from the complexity of computing the Glynn coefficients in Ryser's formula.

We prove the following.

\begin{theorem}
\label{th1}
    For an $n \times n$ matrix $A$ with  entries in $\mathbb{C}$  and $\epsilon>0$, Gurvits algorithm [50] outputs with high probability in $O(\frac{n^2}{\epsilon^2})$-time an estimate $\mathsf{E}(\mathsf{Per}^2(A))$ such that
$$\big ||\mathsf{E}(\mathsf{Per}^2(A))|-|\mathsf{Per}^2(A)| \big| \leq \epsilon(2+ \epsilon) ||A||^{2n}.$$
\end{theorem}
\begin{proof}
    Let $\mathsf{E}(\mathsf{Per}^2(A)):=(\mathsf{E}(\mathsf{Per}(A))^2$, where $\mathsf{E}(\mathsf{Per}(A))$ is the output of Gurvits algorithm. $|\mathsf{E}(\mathsf{Per}^2(A))-\mathsf{Per}^2(A)|=|\mathsf{E}(\mathsf{Per}(A))-\mathsf{Per}(A)||\mathsf{E}(\mathsf{Per}(A))+\mathsf{Per}(A)| \leq \epsilon ||A||^n(|\mathsf{E}(\mathsf{Per}(A))|+|\mathsf{Per}(A)|). $

    Using a reverse triangle inequality
    \begin{equation*}
    \big||\mathsf{E}(\mathsf{Per}(A))|-|\mathsf{Per}(A)|\big| \leq |\mathsf{E}(\mathsf{Per}(A))-\mathsf{Per}(A)| \leq \epsilon ||A||^n,
    \end{equation*}
    and consequently $|\mathsf{E}(\mathsf{Per}(A))| \leq |\mathsf{Per}(A)|+\epsilon ||A||^n.$ Plugging this into the inequality of $|\mathsf{E}(\mathsf{Per}^2(A))-\mathsf{Per}^2(A)|$ gives
    \begin{equation*}
    |\mathsf{E}(\mathsf{Per}^2(A))-\mathsf{Per}^2(A)| \leq \epsilon||A||^n (2|\mathsf{Per}(A)|+\epsilon ||A||^n)  \leq \epsilon(2+\epsilon)||A||^{2n},
    \end{equation*}
    where the rightmost side follows from $|\mathsf{Per}(A)| \leq ||A||^n$. To complete the proof, we again use a reverse triangle inequality 
    \begin{equation*}
    \big ||\mathsf{E}(\mathsf{Per}^2(A))|-|\mathsf{Per}^2(A)| \big| \leq |\mathsf{E}(\mathsf{Per}^2(A))-\mathsf{Per}^2(A)| \leq \epsilon(2+\epsilon)||A||^{2n},
    \end{equation*}
    and this completes the proof.
\end{proof}
 Now, we prove an analogue of the above theorem for estimating $|\mathsf{Per}^2(A)|$ from the output statistics of linear optical circuits.
\begin{theorem}
\label{th2}
Let A be an $n \times n$ matrix with  entries in $\mathbb{C}$, and $\epsilon>0$. There is a linear optical circuit $U$ which, with $O(\frac{1}{\epsilon^2})$ samples and with high probability outputs an estimate $\mathsf{E}(|\mathsf{Per}^2(A)|)$ such that
$$\big|\mathsf{E}(|\mathsf{Per}^2(A)|)-|\mathsf{Per}^2(A)|\big| \leq \epsilon ||A||^{2n}, $$
\end{theorem}
\begin{proof}
Embed $A$ onto a linear optical circuit $U$ using the unitary dilation theorem. From the results of [65], by using Hoeffding's inequality [52], we can with high confidence and with $O(\frac{1}{\epsilon^2})$ samples estimate the probability $p=|\mathsf{Per}^2(A_s)|$ to within precision $\epsilon$, where $A_s:=\frac{A}{||A||}.$ Therefore, we obtain an estimate
$\mathsf{E}(|\mathsf{Per}^2(A_s)|)$ as follows

$$\big|\mathsf{E}(|\mathsf{Per}^2(A_s)|)-|\mathsf{Per}^2(A_s)|\big| \leq \epsilon.$$
  Noting that $|\mathsf{Per}^2(A_s)|=\frac{|\mathsf{Per}^2(A)|}{||A||^{2n}}$, then multiplying both sides of the above inequality by $||A||^{2n}$, and defining $\mathsf{E}(|\mathsf{Per}^2(A)|):=||A||^{2n}\mathsf{E}(|\mathsf{Per}^2(A_s)|)$ completes the proof.
\end{proof}
We note that, as $n$ grows, the probability encoding $\mathsf{Per}(A)$, i.e $p=\frac{|\mathsf{Per}(A)|^2}{||A||^{2n}}$ decreases exponentially with $n$.  This can be seen by noting that $\frac{|\mathsf{Per}(A)|}{||A||^{n}}$ is typically upper bounded  by a function which is decreasing as $e^{-O(n)}$ [31]. This is consistent with the fact that the size of the output space of the boson sampler  grows exponentially with increasing $n$. In the context of Theorem \ref{th2}, this would mean that, in order to asymptotically get  accurate estimates of $p$, $\epsilon$ must decrease exponentially with $n$, and consequently the number of samples collected $O(\frac{1} {\epsilon^2})$ must increase exponentially with $n$. We note that this conclusion about the $\epsilon$ scaling also applies to Theorem \ref{th1}.

\end{document}